%% file: prlqr.tex
\documentclass[12pt,final]{l4dc2021}

\usepackage{float}
\usepackage{graphicx}
\usepackage{times}
\usepackage{siunitx}
\usepackage{paralist}
\usepackage{multicol}

\usepackage{algorithm}
\usepackage{algpseudocode}
\algdef{SE}[DOWHILE]{Do}{doWhile}{\algorithmicdo}[1]{\algorithmicwhile\ #1}%

\newcommand{\ie}{i\/.\/e\/.,\/~}

\newcommand{\cf}{cf\/.\/~}
\newcommand{\fig}{Fig\/.\/~}
\newcommand{\sect}{Sec\/.\/~}

\newcommand{\algo}{Algorithm~}
\newcommand{\theo}{Theorem~}

\newcommand{\est}[1]{\hat #1}
\newcommand{\vect}[1]{\text{vec}(#1)}

\newcommand*\dif{\mathop{}\!\mathrm{d}}
\newcommand{\epsv}{\epsilon_{\text{val}}}
\newcommand{\epspr}{\epsilon_{\text{pr}}}

\newcommand{\sysparams}{[\bar A \; \bar B]}
\newcommand{\PS}{P_{\bar S}}
\newcommand{\PcS}{\smash{P^c_{\bar S}}}
\newcommand{\stable}{\phi}
\newcommand{\kron}{\otimes}

\newtheorem{assumption}{Assumption}

\title[Probabilistic Robust Linear Quadratic Regulators with Gaussian Processes]{Probabilistic Robust Linear Quadratic Regulators \\ with Gaussian Processes}

\usepackage{times}

\author{%
 \Name{Alexander {von Rohr}}$^{1,3,4}$ \Email{vonrohr@dsme.rwth-aachen.de}\\
 \Name{Matthias {Neumann-Brosig}}$^{2,3}$ \Email{matthias.neumann-brosig@iav.de}\\
 \Name{Sebastian Trimpe}$^{1,4}$ \Email{trimpe@dsme.rwth-aachen.de}\\
 \addr $1$ Institut for Data Science in Mechanical Engineering, RWTH Aachen University, Germany\\
 $2$ Institut Analysis und Algebra, TU Braunschweig, Germany \\
 $3$ Ingenieurgesellschaft Auto und Verkehr (IAV) \\
 $4$ Max Planck Institute for Intelligent Systems, T{\"u}bingen/Stuttgart, Germany%
}

\begin{document}

\maketitle

\begin{abstract}%
Probabilistic models such as Gaussian processes (GPs) are powerful tools to learn unknown dynamical systems from data for subsequent use in control design.
While learning-based control has the potential to yield superior performance in demanding applications, robustness to uncertainty remains an important challenge.
Since Bayesian methods quantify uncertainty of the learning results, it is natural to incorporate these uncertainties into a robust design. 
In contrast to most state-of-the-art approaches that consider worst-case estimates, we leverage the learning method's posterior distribution in the controller synthesis. 
The result is a more informed and, thus, more efficient trade-off between performance and robustness.
We present a novel controller synthesis for linearized GP dynamics that yields robust controllers with respect to a probabilistic stability margin.
The formulation is based on a recently proposed algorithm for linear quadratic control synthesis, which we extend by giving probabilistic robustness guarantees in the form of credibility bounds for the system's stability.
Comparisons to existing methods based on worst-case and certainty-equivalence designs reveal superior performance and robustness properties of the proposed method.
\end{abstract}


\input{sections/01_introduction.tex}
\input{sections/02_problem_formulation.tex}

\input{sections/03_framework.tex}
\input{sections/04_results.tex}

\input{sections/05_conclusion.tex}

\newpage
\acks{The authors thank F. Solowjow, C. Fiedler, A. R. Geist and S. Heim for their helpful comments and discussions.
This work was supported in part by the Cyber Valley Initiative and the Max Planck Society.
The authors thank the International Max Planck Research School for Intelligent Systems for supporting A.~von~Rohr.
}

\bibliography{prlqr}

\end{document}

%% file: sections/01_introduction.tex
\section{Introduction}

The trade-off between performance and robustness is at the heart of any practical control design \citep{boulet2007fundamental}: 
do we seek an optimal controller for a particular system \emph{or} one that works well over a larger set of systems and conditions?
It is possible to quantify this trade-off in the various methods of robust control~\citep{zhou1996robust}.
Robust control, however, is typically based on a worst-case treatment of uncertainty sets and does not consider a distribution over the set of possible model parameters.
Control design based on probabilistic models learned from data is less well established.
Yet, for learned models, the performance-robustness trade-off is especially relevant since uncertainty is inherent in any learning result.

When controllers are deployed on physical machines, practical robustness guarantees are essential to ensure the plant's safe operation.
Dealing with large uncertainties, worst-case controllers are often highly conservative and require additional assumptions to retain their guarantees.
Given the learned probabilistic model, it is sometimes impossible to give deterministic stability guarantees if the set of possible model parameters is too large or even infinite.
By relaxing deterministic guarantees and allowing controllers to fail (with low probability), we can achieve a principled robustness-performance trade-off.
This relaxation and the information given by the distribution allow us to improve the expected control performance while retaining practically useful stability guarantees.

Gaussian processes (GPs) are a state-of-the-art method to learn probabilistic dynamics models that represent uncertainty as a distribution over functions.
Even though the use of GPs to learn dynamics is gaining popularity in recent years, subsequent robust controller synthesis based on probabilistic models is less established.
In their pioneering work in this direction, \cite{berkenkamp2015safe} propose a robust controller synthesis for a linearized GP dynamics model.
Therein a worst-case robust synthesis method is presented where the uncertainty sets are derived from the GP posterior variance.

This paper extends prior work by designing controllers that use a linearized GP model's probabilistic uncertainties.
Instead of using the credible region as the bounds for the set of possible system parameters, we leverage the GP's distribution to design efficient controllers with, in effect, the same robustness guarantees but superior performance.
The proposed formulation results in a Bayesian optimal controller design that uses the GP's probabilistic information.   
In particular, the desired properties are (i) probabilistic robustness (\ie stable with high and predefined probability) and (ii) optimal in expectation with respect to (w.r.t.) the posterior distribution.
In summary, the contributions of this paper are:
\begin{compactitem}
  \item a relaxation of the deterministic robust synthesis for GP error bounds proposed by \citet{berkenkamp2015safe} with reduced conservatism and improved control performance;
  \item a synthesis yielding a probabilistically robust LQR for a linearized GP dynamics model;
  \item probabilistic robustness guarantees for the resulting controller, where the desired confidence on the closed-loop stability can be chosen a-priori and is guaranteed by the algorithm.
\end{compactitem}
The method describes an approximately optimal Bayesian robust stochastic control synthesis.
In numerical experiments, we demonstrate improved control performance compared to the worst-case consideration proposed by \cite{berkenkamp2015safe}.

\subsection{Related work}

While GPs are increasingly popular for learning dynamical systems \citep{nguyen2011model,frigola2014variational,svensson2016computationally, geist2020learning,buissonfenet2020actively} and have recently been used in many branches of reinforcement learning \citep{deisenroth2011pilco,doerr2017optimizing,vinogradska2018numerical} and model predictive control \citep{ostafew2016learning, hewing2019cautious, jain2018learning, koller2018learning}, 
only a few approaches exist for offline controller synthesis with formal stability guarantees.
The approach in \citet{berkenkamp2015safe}, as well as other works such as \citet{umlauft2017feedback}, use probabilistic uncertainty to define error bounds, which are then treated in a worst-case fashion, while the underlying distribution is ignored.
A key difference between these prior works and ours is the formulation of the objective:
We demonstrate herein that leveraging the distribution yields better expected performance while retaining Bayesian stability guarantees similar to the deterministic robust synthesis.
In practice, we gain a significant increase in performance at the cost of a small additional but pre-defined probability of being unstable.

\cite{umlauft2018scenario} propose to use a GP dynamics model to sample scenarios for a differential dynamic program.
Their method is applicable to non-linear systems, using an iterative LQR formulation, and includes probabilistic guarantees on the performance.
In contrast to this approach, our method does not require solving the scenario program online but is limited to a single operating point in its current form.
We derive explicit robustness properties for the linear, probabilistic setting.

Our synthesis method is primarily based on an LMI formulation for the upper bound on the LQR problem in the presence of probabilistic uncertainties proposed by \cite{umenberger2018learning}.
This formulation alleviates some of the conservativeness inherent in alternative worst-case formulations and thereby provides a synthesis for the LQR problem in the presence of probabilistic uncertainty.
While the authors show empirically that the controller is robust w.r.t.\ the specified uncertainty, they do not give theoretical stability guarantees for the presented approach.
We extend this work by a theoretical robustness analysis based on results from scenario optimization.

%% file: sections/02_problem_formulation.tex
\section{Problem formulation}\label{sec:problem}

Consider a stochastic, discrete-time and time-invariant dynamical system
\begin{equation}\label{eq:dynamics}
    x_{k+1} = f(x_{k}, u_{k}) + \omega_k,
\end{equation}
and a locally valid linear approximation around an a-priori known operating point $x^* = f(x^*, u^*)$
\begin{equation}\label{eq:linear_dynamics}
    x_{k+1} \approx \bar A x_{k} + \bar B u_{k} + \omega_k,
\end{equation}
where $x_k \in \mathbb{R}^{d_x}$ is the state, $u_k \in \mathbb{R}^{d_u}$ is the input and $\omega_k \sim \mathcal{N}(0,\Sigma^\omega)$ is the process noise of the system. 
For notational convenience, we define the state-action tuple $q_k \, \dot= \, (x_{k}, u_{k})$ as well as the parameters of the linearized system as $\bar S = \sysparams$.
We assume access to the state $x_k$ and the process noise is i.i.d. with $\Sigma_\omega = \textup{diag} (\sigma_{\omega,1}^2,\dots, \sigma_{\omega,d_x}^2)$.

Designing controllers for a-priori known operating points that have a locally valid linear approximation of their dynamics is a problem often considered in practice, as well as in \citet{berkenkamp2015safe} and many industrial control processes.

\subsection{Gaussian process dynamics model}

The dynamics $f$ and the parameters of the linear approximations $\bar A$ and $\bar B$ are unknown, but we can formulate a prior $\hat f$ on \eqref{eq:dynamics} in the form of a GP.
After observing a set $D_N = \{q_{k}, x_{k+1}  \}_{k=0}^N$ of one step transitions (usually as trajectories) of system \eqref{eq:dynamics}, we compute the posterior for the dynamics $f$ as well as the linearized dynamics $\bar A$ and $\bar B$.
We get a Bayesian estimate for the system matrices in \eqref{eq:linear_dynamics} by taking the partial derivatives of the GP $\est{f}$. 
The linearization of the GP dynamics model requires that the covariance function is at least twice differentiable at the operating point.
Since GPs are closed under linear transformations, the resulting Jacobian at the point $q^*$ is distributed according to a matrix normal distribution (\cf \citet[Chapter~9.4]{rasmussen2006gaussian})
$\mathcal{MN}(M_S, U, V)$
where $M_S$ is the mean and $U$ and $V$ are the covariances between rows and columns.%
\footnote{The matrix normal distribution and multivariate normal distribution are related by $\vect{\est{S}} \sim \mathcal{N}(\vect{M_S}, U \otimes V)$, where $\vect{M}$ is the (column-major) vectorization of $M$ and $\kron$ is the Kronecker product.}
We denote the posterior probability distribution over system parameters as $\PS$ and its support as $\Pi$.
For the synthesis problem, we define a truncated matrix normal distribution $\PcS$ with the same mean and covariances, and its support $\Pi_c$ with 
$$
\Pi_c = \{S \; | \; (\vect{S} - \vect{M_S})^{T} (U \otimes V)^{-1} (\vect{S} - \vect{M_S}) \leq \mathcal{X}^2_{d_s, c}\},
$$
where $\mathcal{X}^2_{d_s, c}$ is the c-quantile of the Chi-squared distribution with $d_s = d_x^2 + d_x d_u$ degrees of freedom.
\begin{remark}
Learning a non-linear GP model for \eqref{eq:dynamics} instead of directly estimating a linear model for \eqref{eq:linear_dynamics} using, for example, Bayesian linear regression, can be advantageous since the data $D_N$ is collected from the non-linear system. 
The linear model will be biased if the data is collected outside the local region where \eqref{eq:linear_dynamics} is a good approximation.
Modeling $f$ directly can avoid this bias.
\end{remark}

\subsection{Bayesian optimal linear quadratic regulator}

We want to design a linear and static state-feedback controller $u_k = K x_k$ for \eqref{eq:linear_dynamics} that is optimal, in expectation, w.r.t. the truncated posterior distribution given by the GP model
\begin{subequations}\label{eq:cost_all}
\renewcommand{\theequation}{\theparentequation.\arabic{equation}}
\begin{align}
\min_K \quad J(K)& = \int_{\Pi_c} J(K, S) \PcS(S) \dif S, \label{eq:cost} \\
J(K,S)& = \lim_{T\rightarrow \infty} \frac{1}{T} \sum^T_{k=0} E_{\omega}\left[ x_k^T Q x_k  + u_k^T R u_k\right] \label{eq:cost_state} \\
\text{s.t.} \quad x_{k+1} &= (A + BK) x_{k} + \omega_k, \; S = [A \; B], \nonumber
\end{align}
\end{subequations}
where $Q$ and $R$ are user defined, positive semi-definite and positive definite weight matrices respectively and $x_0 = x^*$.
This formulation is different from the standard LQR problem with known $A$ and $B$ matrices. 
The case where uncertainty about $\bar A$ and $\bar B$ is small and can be ignored, certainty equivalence (CE), is investigated by \citet{mania2019certainty}.
Due to optimization of the \emph{expected cost} w.r.t. the posterior belief the Bayesian optimal formulation may yield a different controller, even if the CE controller can stabilize all systems in $\Pi_c$ (see \sect\ref{sec:results}).

In \eqref{eq:cost_state} we consider the (usual) expectation over the states and time, given a fixed system.
Additionally, we take the expectation over model uncertainty \eqref{eq:cost}.
If \eqref{eq:cost} exists for some $K$, this $K$ will almost surely stabilize any system in $\Pi_c$, since any unstable systems will have infinite expected cost.
Conversely, if \eqref{eq:cost_all} is infeasible, there exists no $K$ that is robustly stable.
\begin{assumption}\label{as:finite}
There exists at least one $K$ such that \eqref{eq:cost} is finite on the support of $\PcS$.
\end{assumption}
Assumption~\ref{as:finite} excludes any distributions containing systems that cannot be stabilized with the same controller.
Distributions with unbounded support, such as the non-truncated normal distribution, are excluded since these will always contain unstable systems for every $K$.
As a practical way of dealing with the unbounded support of a GP, we approximate $\PS$ by a truncated distribution $\PcS$ over a bounded subset given by the credible interval $c$.

Depending on the prior distribution, likelihood, and data, the support of $\PcS$ might contain uncontrollable/unstabilizable systems.
Non-controllable pairs of $A$, $B$ are sets of measure zero w.r.t. the Lebesgue measure \citep[Prop. 3.3.12]{sontag1998mathematical}, meaning a sample of the GP linearization is almost surely controllable.
Assumption~\ref{as:finite} also excludes distributions that concentrate around an uncontrollable system as all of these systems are \emph{hard} to control, \ie the controllability matrices have small singular values, which will lead to exploding costs for every $K$.

%% file: sections/03_framework.tex
\section{Synthesis for linearized Gaussian processes}\label{sec:synthesis}
In this section, we introduce the synthesis algorithm by which we approximately optimize \eqref{eq:cost}, and we prove a probabilistic stability property for the resulting controller.

\subsection{Controller synthesis}

To find a controller that minimizes the expected quadratic cost, it would be necessary to solve the integral in $\eqref{eq:cost}$, for which no analytic solution exists \citep{umenberger2018learning}.
\cite{umenberger2018learning} developed a synthesis algorithm that optimizes a tight and convex upper bound for a Monte Carlo approximation to \eqref{eq:cost_all} by improving an initial common Lyapunov solution using a majorize-minimize algorithm.

Since the chosen GP model yields a point estimate for the process noise $\Sigma^\omega$ we can slightly modify the LMI formulation of \cite{umenberger2018learning},

{
\noindent\begin{minipage}{.5\linewidth}
    \begin{equation}\label{eq:init}
    \begin{aligned}
    &\min_{L, Z\in\mathbb{S}, Y\in\mathbb{S}_{++}} \quad \text{tr} \quad Z \\
    \text{s.t.} \quad &\begin{bmatrix}
        Y & * & * & * \\
        A_i Y + B_i L & Y & * & * \\
        Q^{1/2} Y & 0 & I & * \\
        L & 0 & 0 & R^{-1}
    \end{bmatrix} \succcurlyeq 0 \\
    &\begin{bmatrix}
        Z & {\Sigma^\omega}^{1/2} \\
        {\Sigma^\omega}^{1/2} & Y
    \end{bmatrix}  \succcurlyeq 0 \\
    &\forall \quad S_i = [A_i \; B_i] \in \Pi_c^M, \\
    \end{aligned}
    \end{equation}
\end{minipage}%
\begin{minipage}{.5\linewidth}
    \begin{equation}\label{eq:improv}
    \begin{aligned}
    &\min_{K, {X_i}\in\mathbb{S}_{++}} \quad \text{tr} \quad \frac{1}{M} \sum_i X_i {\Sigma^\omega} \\
    &\text{s.t.} \quad \begin{bmatrix}
        X_i - Q & * & * \\
        A_i + B_i K & T_{\bar X_{i}}(X_i) & * \\
        K & 0 & R^{-1}
    \end{bmatrix} \succcurlyeq 0 \\
    &\forall \quad S_i = [A_i \; B_i] \in \Pi_c^M, \\
    \end{aligned}
    \end{equation}
\end{minipage}
\smallskip
}

\noindent
where $\mathbb{S}$ and $\mathbb{S}_{++}$ denote the set of symmetric or positive definite matrices.
For the initial solution \eqref{eq:init} $X=Y^{-1}$, $K = L X$, and $Z$ is a slack variable.
For the improvement step \eqref{eq:improv} $T_{\bar X_{i}}(X_i)$ denotes the linear approximation of $X_i^{-1}$ around $\bar X_{i}$, the solution to \eqref{eq:improv} for the previous iteration.

Let $K^*$ denote the solution of \eqref{eq:improv} after a stopping criterion has been fulfilled.
We are interested in the probability that $K^*$ will stabilize the true linear system $x_k = (\bar A + \bar B K^*)x_k + w_k$.
As a first step we establish an \emph{a-priori} bound on the probability that the controller result of the initialization step \eqref{eq:init} $K_{\text{init}}=LX$ will be stable w.r.t. to the set of probable systems $\Pi_c$.
We use the probabilistic knowledge about the system parameters to express the system's stability as a Bernoulli distribution.
We define a stability indicator variable for the true linear system \eqref{eq:linear_dynamics} and a fixed $K$ as
\begin{equation}
\stable_K = \begin{cases}
                1, & \text{if}\ \rho(\bar A + \bar B K) < 1  \\
                0, & \text{otherwise.}
            \end{cases}
\end{equation}
Further, we define the probabilistic robustness property as $V_K(\epsilon) = \left( \mathbb{E}[\stable_K] \geq 1 - \epsilon \right)$.
By this definition, we accept a (small) probability $\epsilon$ that the feedback control destabilizes the system.
\setcounter{theorem}{0} 
\begin{theorem}\label{th:scenario}
Let Assumption~\ref{as:finite} hold and let $K_{\text{init}}$ be the solution to \eqref{eq:init}.
Let $\epsilon$ and $\beta >0$ be given and $n_k = 2 d_x^2 + d_x d_u$ the number of optimization variables in \eqref{eq:init}. Let $\PcS$ be the posterior probability over $S = \sysparams$. 
If we draw
$M \geq \left\lceil \frac{2}{\epsilon} \ln (\frac{1}{\beta}) + n_k \right\rceil$
samples according to $\PcS$, then it holds that
\begin{equation}\label{eq:result_synth}
\begin{aligned}
    P( V_{K_{\text{init}}}(\epsilon) ) > 1 - \beta.
\end{aligned}
\end{equation}
\end{theorem}
\begin{proof}
The optimization problem \eqref{eq:init} is a convex scenario program and therefore we can apply \citet[Theorem 1]{campi2009scenario}.
\end{proof}
Theorem~\ref{th:scenario} states that, according to our current (truncated) belief over the system parameters, if we sample $M$ scenarios from $\PcS$, we can be \emph{confident} that with probability of at least $1-\beta$ the desired robustness property $V_{K_{\text{init}}}$ is fullfilled.
Part of the robustness property is a \emph{risk} parameter that describes the upper bound for the probability that $K_{\text{init}}$ will destabilize the system.

Since $K_{\text{init}}$ is highly conservative due to the common Lyapunov approximation, we improve the controller by iteratively solving \eqref{eq:improv}.
Due to the averaging over all solutions, \eqref{eq:improv} is not a scenario program and Theorem~\ref{th:scenario} does not hold for the improvements.
We therefore validate the final controller in an a-posteriori analysis which we describe in the next subsection.

\subsection{Controller validation}

After iteratively solving \eqref{eq:improv} for a particular set of samples, we validate the success of synthesis, \ie that the controller will stabilize \eqref{eq:linear_dynamics} with high probability, \emph{before} deploying the controller to the system.
The validation ensures that the robustness properties that provably hold for the initialization \eqref{eq:init} will also hold after the iterative improvement.
At the same time we can improve our initial guarantees regarding the probability of a successful synthesis $P( V_{K^*}(\epsilon))$. 

The validation is done on the learned distribution but for samples not seen by the synthesis method.
For the validation procedure, we approximate the expectation $\mathbb{E}[\stable_{K^*}]$ by sampling systems from $\PcS$.
The sample-based approach avoids conservativeness in our analysis at the cost of some approximation error.
Since sampling and checking the spectral radius are computationally relatively cheap, we can use a lot of samples and make this approximation error small enough for most practical purposes.
Sample-based robust analysis is a simple method that does not rely on a specific uncertainty structure (\cf \citet{ray1993monte}). 
\begin{theorem}\label{th:validation}
Let Assumption~\ref{as:finite} hold. 
If we draw $M_{\text{val}} \geq \frac{1}{2 \epsv} \log(\frac{1}{\alpha})$ samples from $\PcS$, \algo\ref{alg:synth} will return a probabilistically stable controller $K^*$ where
\begin{equation}\label{eq:result_val}
\begin{aligned}
    P\left( \mathbb{E}[\stable_{K^*}] \geq 1 - \epspr) \right) > 1 - \alpha,
\end{aligned}
\end{equation}
with $\epspr =  c - (\epsilon + \epsv)$ and according to the posterior distribution $\PS$.
\end{theorem}
\begin{proof}
\algo\ref{alg:synth} returns a controller only if the empirical expected value of $\stable$, $\est{\stable}$ is at least $1-\epsilon$.
We can directly apply the one sided Hoeffding inequality $P(\est{\stable} - \mathbb{E}[\stable] \geq \epsv) \leq \exp{(-2 M_{\text{val}} \epsv^2)}$.
By choosing $M_{\text{val}} \geq \frac{1}{2 \epsv} \log(\frac{1}{\alpha})$ it follows that $\exp{(-2 M_{\text{val}} \epsv^2)} \leq \alpha$ and $\est{\stable} \geq 1-\epsilon$, thus
\begin{equation}
\begin{aligned}
P(1 - \epsilon - \mathbb{E}[\stable] \geq \epsv) \leq \alpha
\Rightarrow P(\mathbb{E}[\stable] \leq 1 - \epsilon - \epsv) \leq \alpha. \\ 
\end{aligned}
\end{equation}
Since we sample scenarios from the truncated distribution $\PcS$ as opposed to $\PS$, we do not account for systems in a set with probability mass $1-c$.
It follows that $P\left( \mathbb{E}[\stable] \leq \epspr \right) \leq \alpha$.
\end{proof}
If Assumption~\ref{as:finite} does not hold, \algo\ref{alg:synth} will return no controller with high probability, giving a certificate of infeasibility for the problem approximated by the scenario program.
In this case, we can either improve the prior distribution by adding additional knowledge or collect more data.
The choice of $\beta$ in \algo\ref{alg:synth} does not influence the a-posteriori stability, but it determines how often we need to run the synthesis on a new sample set before we find a controller that satisfies $\est{\stable} \geq 1-\epsilon$.
Consequently choosing a high $\beta$ will increase the expected runtime of \algo\ref{alg:synth}.

%

When compared to the synthesis by \cite{berkenkamp2015safe} the problem is relaxed since we allow a (small) probability $\alpha$ that the controller does not robustly stabilize the credible region of the model. 
Otherwise, we can achieve the same stability guarantees by setting $\epspr = 95\%$, the credible region chosen therein.
Keeping the similar robustness properties in mind, we now compare both approaches in the next section.

\begin{algorithm}
	\caption{Probabilistic robust LQR synthesis}\label{alg:synth}
	\begin{algorithmic}[1]
		\State \textbf{Input:} Posterior distribution $\PcS$ of $\est{S}$, risk parameter $\epsilon$ and $\epsv$, confidence parameter $\beta$ and $\alpha$, cost matrices $Q$ and $R$.
          \Do
			\State Generate at least $M = \left\lceil \frac{2}{\epsilon} \ln (\frac{1}{\beta}) + n_k \right\rceil$ samples from $\PcS$ (Theorem~\ref{th:scenario}).
    		\State Find $K^*$ by solving first \eqref{eq:init} and than \eqref{eq:improv} iteratively.
    		\State \textbf{if} {\eqref{eq:init} is infeasible} \Return
		    \State Draw at least $M_{\text{val}} = \left\lceil \frac{1}{2 \epsv} \log(\frac{1}{\alpha}) \right\rceil$ samples from $\PcS$. (Theorem~\ref{th:validation})
		    \State $
		    \est{\stable_{K^*}} = \frac{1}{M_{\text{val}}} \sum_{i=0}^{M_{\text{val}}}
		    \begin{cases}
                1, & \text{if}\ \rho(A_i + B_i K^*) < 1  \\
                0, & \text{otherwise}
            \end{cases}
            $
          \doWhile{$\est{\stable_{K^*}} < 1 - \epsilon$}
		\State\Return $K^*$.
	\end{algorithmic}
\end{algorithm}

%% file: sections/04_results.tex
\section{Empirical results}\label{sec:results}

In this section, we apply the results from \sect\ref{sec:synthesis} to a synthetic benchmark problem to illustrate and empirically verify the performance of the proposed synthesis and its robustness properties.
We compare performance and robustness of the proposed \emph{probabilistic robust} (PR) synthesis described in \algo\ref{alg:synth} with the \emph{robust} (R) method to solve the LQR problem based on linearized GP models in \citet{berkenkamp2015safe}.
Additionally, we compare against the \emph{certainty equivalence} (CE) setting in which the mean of $\PS$ is used to solve the standard Riccati equations. 

In the second subsection, we use GPs to learn a given system from data and synthesize a controller based on the distribution over linearized systems.
We numerically compare the stability and performance properties of the discussed methods.
We want to remark that the primary goal here is not to evaluate how well GPs learn the given dynamics.
Instead, our focus is to compare how different approaches use the given uncertainty and how the resulting controller performs. 

We collect the confidence and risk parameters for Algorithm~\ref{alg:synth} in the tuple $\psi = (c, \epsilon, \beta, \epsv, \alpha)$. 
The source code for all experiments, including all parameters, is available online\footnote{\url{https://github.com/Data-Science-in-Mechanical-Engineering/prlqr}}.

\subsection{Empirical results on a synthetic distribution}\label{ssec:results_dist}

As a benchmark problem, we create a synthetic distribution $\PS$ as $\mathcal{MN}([\bar A \; \bar B], U, V)$ that highlights the benefits of the proposed method.
As the covariance matrices we choose a sample $E = U \kron V$ from a Wishart distribution with $(d_x^2 + d_x d_u)$ degrees of freedom and a scale matrix $\sigma^2  (0.5 \cdot  I_{18} + 0.5 \cdot  1_{18})$, where $1_{18}$ denotes a $18\times18$ all-ones matrix.
This allows for correlations in the uncertainty between system parameters.
By changing $\sigma^2$ we can compare the synthesis methods over multiple levels of uncertainty.

As the mean for the synthetic distribution we choose the LQR benchmark problem proposed by \cite{dean2020sample}
\begin{equation}\label{eq:dean_dynamics}
\bar A =
\begin{bmatrix}
    1.01 & 0.01 & 0\\
    0.01 & 1.01 & 0.01\\
    0 & 0.01 & 1.01
\end{bmatrix} \;
\bar B = I_3 \; Q = 10^{-3}I_3 \; R=I_3 \; \Sigma^\omega=10^{-3}I_3,
\end{equation}
where $I_3$ denotes the 3-dimensional identity matrix.
The credible region for the robust synthesis is chosen as $95\%$ and the parameters for the proposed method are set correspondingly to $\psi = (0.98, 0.02, 0.20, 0.01, 0.001)$.
While the distribution of risk over the parameters $\epsilon_{\text{pr}} = c - (\epsilon + \epsv)$ is arbitrary, it does influence the numbers of samples for the scenario program as well as the validation.
In practice, we set the parameters via trial and error, such that the run-time of the algorithm is reasonable.
Stability and performance are evaluated on a sample of \num{10.000} systems drawn from the \emph{non-truncated} matrix normal distribution.
We calculate the expected cost for each sampled system over a finite horizon of $200$ (\cf \citet[Eq\/.\/~20]{schluter2020event}).
We report the cost for all sampled and stable closed-loop systems as well as the frequency of unstable systems in \fig\ref{fig:synth_results}.

The empirical results clearly show a reduced level of conservativeness in the problem formulation when comparing the proposed PR synthesis to the R synthesis.
The PR synthesis problem is feasible with higher levels of uncertainty in the model and, additionally, the \emph{expected cost} of the resulting controller is lower for all levels of uncertainty.
The trade-off is that the PR synthesis yielding an unstable controller for up to $0.3\%$ of samples of the \emph{non-truncated} distribution.
As to be expected, having the different optimization goals in mind, the worst case costs for the controller of the PR synthesis is higher, but only for high levels of uncertainty and in rare cases.
Similar empirical results for the PR synthesis have been shown previously by \cite{umenberger2018learning}, albeit without sample requirements for probabilistic robustness guarantees.
In contrast to the previous results we increased the sample size to $M \geq 188$ according to \theo\ref{th:scenario} from $M=50$ used by \cite{umenberger2018learning}.


Compared to the robust synthesis methods, the CE synthesis yields a controller that stabilizes at least $95\%$ of the samples for very low uncertainties in the system parameters. 
It can, therefore, not be used in settings with significant uncertainties where we require probabilistic performance guarantees.
Even at low model uncertainty, the cost's variance of the CE controller remains high, which might be undesirable in practice. 
Starting from a $\sigma^2$ of $1\text{e-}4$, the controller resulting from the PR synthesis is able to outperform the CE controller in expectation.
This suggests that robustness properties do not need to come at the cost of performance in the considered setting with probabilistic uncertainties.
To summarize the findings of \fig\ref{fig:synth_results}: the proposed \algo\ref{alg:synth} outperforms the robust controller in expectation, has higher worst-case costs only in very unlikely cases, and can find controllers with significantly higher model uncertainty.

\newcommand{\figSynth}[1]{
\begin{figure*}[#1]
\centering
\includegraphics[width=\textwidth]{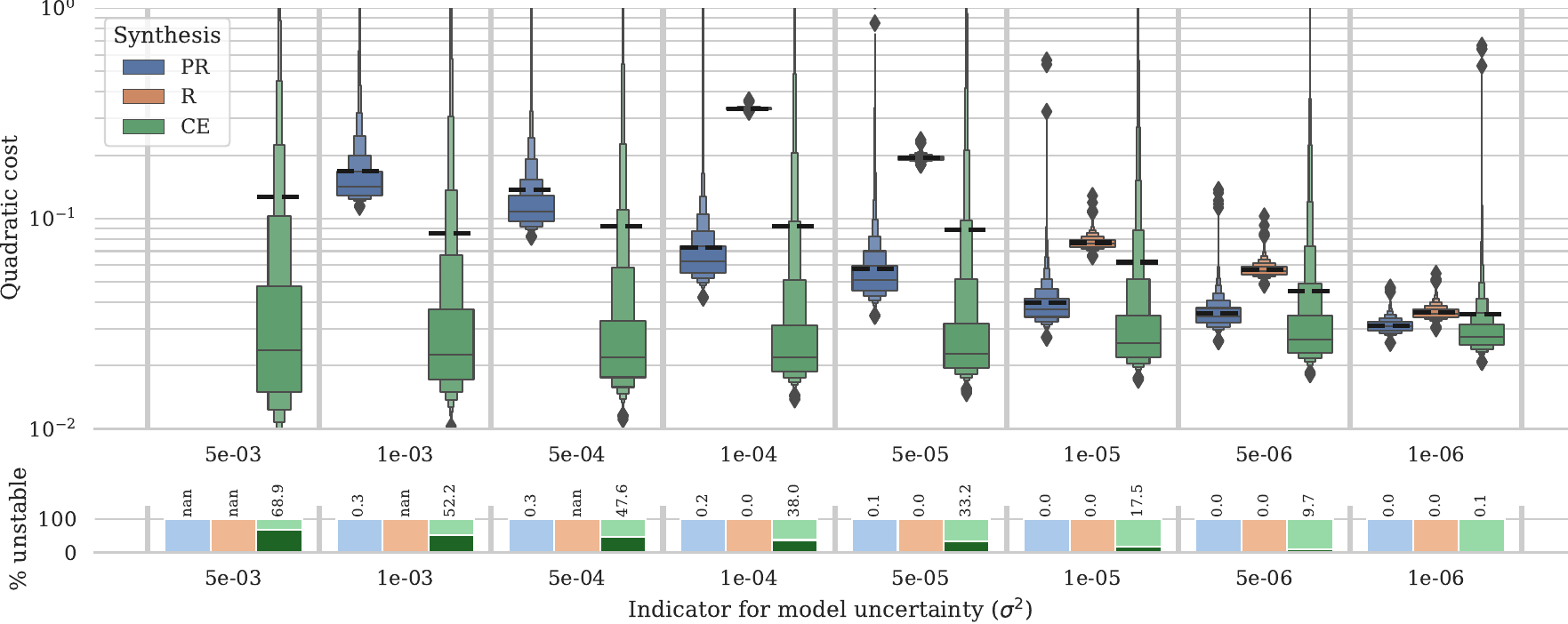}
\caption{Synthetic example (Sec.~\ref{ssec:results_dist}).  Distribution over LQR costs (top)  and frequency of unstable closed-loop systems (bottom) for a model with different levels of uncertainties based on \num{10.000} samples.
A missing boxplot indicates the problem was infeasible for the given uncertainty.
The empirical mean of each distribution is shown as a black dashed line.
}
\label{fig:synth_results}
\end{figure*}
\vspace{-2.5mm}
}
\figSynth{htp}

\subsection{Synthesis using GP models trained in simulation}\label{ssec:results_gp}

In this subsection, we apply our framework to stabilize two non-linear systems, namely the \emph{Furuta} or \emph{rotary pendulum} and a \emph{synthetic} system inspired by the problem proposed in \cite{dean2020sample}.
As a comparison baseline, we show the controller based on the linearization of the true system (T).
We add this comparison to show how the learning-based methods perform versus \emph{perfect model knowledge}.
The parameters $\psi$ for the PR synthesis are identical to the previous subsection with $\psi = (0.98, 0.02, 0.20, 0.01, 0.001)$ and the credible region for the robust method is again set to $\num{95}\%$.
We use a standard GP with a squared exponential kernel and fixed hyperparameters for all our experiments.
A dataset $D_N$ for the GP model is generated by starting each system close to the operating point and sampling random control inputs from a normal distribution.
The system is reset after $n$ samples, and $N_t$ rollouts are recorded, giving a total sample size of~${N=n N_t}$.
We evaluate the \emph{performance} by applying the controller to the true non-linear system, initialized in the operating point, and record the cost for \num{200} time steps. 
We repeat the performance measurement \num{200} times and report the mean of these runs as the cost for the given state feedback controller.
Each experiment, \ie data collection, GP regression, and controller synthesis, is repeated with a different random seed 25 times, and the results for each experiment are plotted.

\subsubsection{Synthetic system}

The locally linear synthetic system is a variant of the linear system proposed as a challenging LQR benchmark by \cite{dean2020sample}
\begin{equation}\label{eq:dean_dynamics_non_lin}
    x_{k+1} =
    \begin{bmatrix}
        1.01 & 0.01 & 0\\
        0.01 & 1.01 & 0.01\\
        0 & 0.01 & 1.01
    \end{bmatrix}
    x_{k}
    +
    \left(  
    \begin{bmatrix}
        0.3 & 0.0 & 0.0 \\
        0.3 & 0.3 & 0.0 \\
        0.3 & 0.3 & 0.3
    \end{bmatrix}
    x_{k}\right)^{\circ 3}
     + I u_{k} + \omega_k,
\end{equation}
were $\circ$ denotes the Hadamard power.
We set $Q=I_3$ and $R=I_3$.
One rollout for this system consists of six samples.
Due to the added non-linearity, a linear estimator might fail to identify the system.
We evaluate the stability by checking the spectral radius w.r.t. the linear part of \eqref{eq:dean_dynamics_non_lin} and check that all states are smaller than \num{1000}.
The results are shown in \fig\ref{fig:result_synth}. 
The PR synthesis outperforms the R synthesis in expectation and is feasible more often at 3 and 5 rollouts. 
All feasible cases lead to a stable controller.
In this experiment, the CE controller performs better than the robust methods, suggesting that our GP model might overestimate the uncertainty.
For 8 rollouts the PR and CE controller are on par with a controller designed for the linearization of the true non-linear system (T).

\subsubsection{Rotary Pendulum}
The rotary pendulum dynamics are based on the equations of motion derived by \cite{iwashiro1996energy} and discretized with a sample time of $20 \, \text{ms}$.
The state-space of the rotary pendulum is \mbox{4-dimensional} and consists of the angle and angular velocity of the rotary arm and the pendulum. The operating point is the unstable upright position.
We collected \num{30} samples per rollout.
Since our contribution is not focused on model learning, we reduced the difficulty of the learning task.
The data is recorded around the operating point using an a-priori stabilizing controller, which allowed us to have a relatively accurate GP posterior with comparatively little data.

The results are shown in \fig\ref{fig:result_furuta}.
The PR synthesis yields a controller that performs better, on average, than both the controllers from the CE and R synthesis.
In this setting, the PR controller performs better than the CE controller, which we attribute to the more informative data-set and a better tuned GP model, in contrast to the synthetic system.

We want to remark that if we use a Bayesian uncertainty estimate, the performance critically depends on the chosen prior, especially in the low data regime.
However, the robustness properties hold as long as the prior contains the true system with sufficient probability.
In principle, the same results apply to uncertainties obtained from frequentist methods.

\begin{figure}[htp]
\floatconts
  {fig:gp_results}
  {\caption{Synthesis on GP model learned from data \sect\ref{ssec:results_gp}. Distribution over (stable) LQR costs (top) on the true system. We plot the frequency of feasibility for each method (bottom).
  All feasible problems resulted in stable controllers.
  We conducted \num{25} experiments for each number of rollouts.}}
  {%
    \subfigure[Synthetic system]{\label{fig:result_synth}%
      \includegraphics[width=0.47\textwidth]{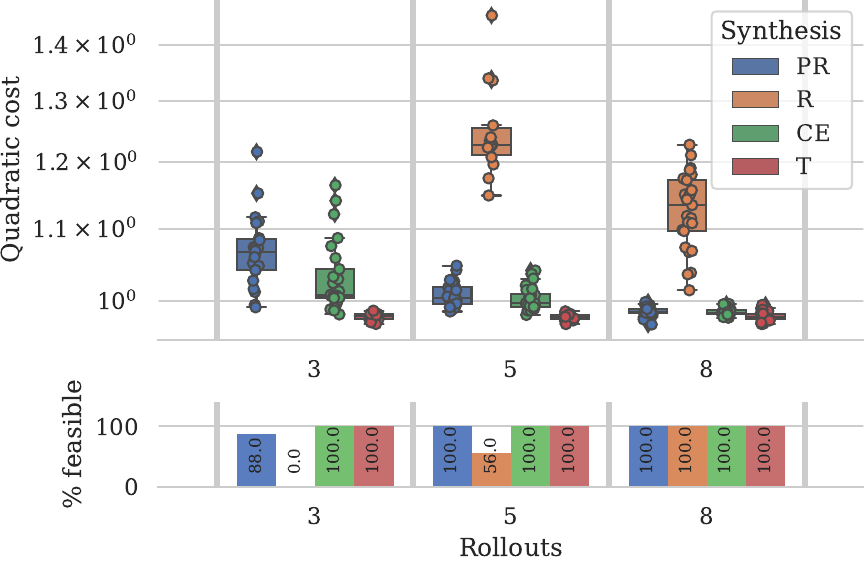}}%
    \qquad
    \subfigure[Rotary pendulum simulation]{\label{fig:result_furuta}%
      \includegraphics[width=0.47\textwidth]{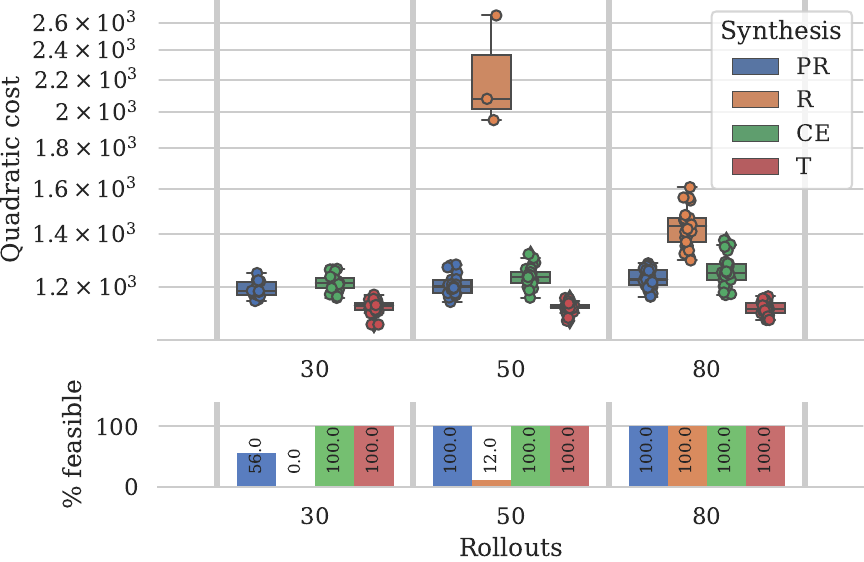}}%
  }
  \vspace{-2.5mm}
\end{figure}

%% file: sections/05_conclusion.tex
\section{Conclusion}

We introduce a novel probabilistic robust control method for linearized GP models that leverages posterior distributions over model parameters.
Instead of optimizing for the worst-case, we optimize the controller for the specific distribution. 
Thus, we achieve significantly better performance while still ensuring comparable robustness guarantees.

From a theoretical point of view, our method only requires a posterior distribution over potential system parameters.
Here, we obtain these by linearizing a GP, a model class with a lot of recent attention for dynamics modeling.
While this is already relevant for many applications, we want to address more general problems, also beyond set-point control, in future work.
Additional research is required to use the framework of probabilistic controller synthesis for partial state observations, where one would like to design an observer alongside the controller, both based on data.